\documentclass[a4paper]{article}
\usepackage[english]{babel}
\usepackage[utf8]{inputenc}
\usepackage{amsmath}
\usepackage{graphicx}
\usepackage[colorinlistoftodos]{todonotes}

\usepackage{amsthm}
\usepackage{amsfonts}
\usepackage{amssymb}
\usepackage{amscd}
\usepackage{stmaryrd}
\usepackage{float}
\usepackage[all]{xy}
\usepackage{relsize}
\usepackage{geometry}
\usepackage{subcaption}

\usepackage{natbib}
\setcitestyle{numbers,open={[},close={]}}

\usepackage[noend]{algpseudocode}
\usepackage{algorithm,algorithmicx}

\algnewcommand\algorithmicforeach{\textbf{for each}}
\algdef{S}[FOR]{ForEach}[1]{\algorithmicforeach\ #1\ \algorithmicdo}

\theoremstyle{plain}
\newtheorem{theorem}{Theorem}[section]

\newtheorem{lemma}[theorem]{Lemma}

\theoremstyle{definition}

\newcommand{\lemref}[1]{Lemma \ref{#1}}

\newcommand{\PP}{\ensuremath \mathbb{P}}

\newcommand{\C}{\ensuremath \mathcal{C}}

\usepackage{algpseudocode,algorithm,algorithmicx}
\algdef{S}[FOR]{ForEach}[1]{\algorithmicforeach\ #1\ \algorithmicdo}

\title{Efficient Correlation Clustering Methods for Large Consensus Clustering Instances}

\author{Nathan Cordner and George Kollios}

\date{July 2023}

\begin{document}

\graphicspath{{figures/}}

\newgeometry{top = 1 in, left = 1 in, right = 1 in, bottom = 1 in}

\maketitle

\begin{abstract}

Consensus clustering (or clustering aggregation) inputs $k$ partitions of a given ground set $V$, and seeks to create a single partition that minimizes disagreement with all input partitions. State-of-the-art algorithms for consensus clustering are based on correlation clustering methods like the popular Pivot algorithm. Unfortunately these methods have not proved to be practical for consensus clustering instances where either $k$ or $V$ gets large.  

In this paper we provide practical run time improvements for correlation clustering solvers when $V$ is large. We reduce the time complexity of Pivot from $O(|V|^2 k)$ to $O(|V| k)$, and its space complexity from $O(|V|^2)$ to $O(|V| k)$---a significant savings since in practice $k$ is much less than $|V|$. We also analyze a sampling method for these algorithms when $k$ is large, bridging the gap between running Pivot on the full set of input partitions (an expected 1.57-approximation) and choosing a single input partition at random (an expected 2-approximation). We show experimentally that algorithms like Pivot do obtain quality clustering results in practice even on small samples of input partitions.

\end{abstract}

\section{Introduction}

In this paper we examine the consensus clustering problem (also known as clustering aggregation). Given a set of input clusterings over a single ground set, consensus clustering seeks to create a single clustering that is most ``similar'' to the inputs provided. For example, a number of organizers for a dinner party may have come up with separate seating arrangements for guests; consensus clustering would find a ``consensus'' seating arrangement that minimizes ``disagreements'' (pairs of guests who are seated at the same table in one arrangement, but at different tables in the other) between the consensus arrangement and each of the arrangements from party organizers. 

Consensus clustering is closely related to the correlation clustering problem. As originally defined by Bansal et al. \cite{bansal2004correlation}, ``min disagreement'' correlation clustering inputs a complete graph $G = (V, E)$ where every pair of nodes is assigned a positive (+) or negative (-) relationship. The objective is to cluster together positively related nodes and separate negatively related ones, minimizing the total number of mistakes along the way. This clustering paradigm has been used in many applications, such has its original motivation of classification \cite{bansal2004correlation}, database deduplication \cite{haruna2018hybrid}, and community detection in social networks \cite{veldt2018correlation,shi2021scalable}. This formulation of graph clustering has been especially useful, since a specific number of clusters does not need to be specified beforehand and the only information needed as input concerns the relationship between objects---not about the objects themselves.

Consensus clustering can be easily reduced to a generalized version of correlation clustering where graph edges are no longer just positive or negative, but are assigned weights (or probabilities) between 0 and 1 \cite{ailon2008aggregating}. The state-of-the-art consensus clustering algorithm is based on the Pivot method for correlation clustering by Ailon et al. \cite{ailon2008aggregating}, which yields a 1.57-approximation result for consensus clustering. One bottleneck that previous authors have run into when applying the Pivot algorithm to consensus clustering involves computing and storing edge weights \cite{goder2008consensus}, making it prohibitive to perform consensus clustering when the number nodes in the graph is large. In Section \ref{sec:consensus_onthefly} we show a practical run time improvement that enables Pivot and other correlation clustering algorithms to run on these instances. In particular we reduce the running time of a single run of Pivot from $O(|V|^2 k)$ to $O(|V| k)$, where $V$ is the ground set and $k$ is the number of input clusterings. We also reduce the amount of memory needed from $O(|V|^2)$ to $O(|V| k)$.

Another bottleneck experienced when running correlation clustering algorithms for consensus clustering is when the number of input clusterings is large. In Section \ref{sec:consensus_sampling} we analyze the effect of computing a consensus on small samples of input clusterings, and show that correlation clustering algorithms like Pivot still produce quality results. In particular we develop a function that computes the expected approximation bound guaranteed by the Pivot algorithm when performing consensus clustering on a sample of input clusterings, filling the gap between using the full set (an expected 1.57-approximation) and choosing a single input clustering at random (known to be an expected 2-approximation).

We conclude this paper with several experimental results that demonstrate the practicality of our new methods for consensus clustering (see Section \ref{sec:consensus_experiments}).

% Finally, in Section 6 we provide several experiments that compare the results of Pivot and Local Search for correlation and consensus clustering, showing that the Inner Local Search algorithm is scalable for large network graphs and in many cases improves the quality of the results nearly as much as Local Search.

\subsection{Related Work}

\subsubsection{Correlation Clustering}

The NP-hard correlation clustering problem was introduced by Bansal et al. \cite{bansal2004correlation}, who also provided its first constant approximation algorithm in the min disagreement setting. The best known approximation factor is $1.994 + \epsilon$, from a linear program rounding method due to Cohen et al. \cite{cohen2022correlation}. Correlation clustering remains an active area of research, and many variations of the problem have arisen over time; a general introduction to the correlation clustering problem and some of its early variants is given by Bonchi et al. \cite{bonchi2014correlation}.

\subsubsection{Weighted Correlation Clustering}

Weighted correlation clustering was first considered by Bansal et al. \cite{bansal2004correlation}. Ailon et al. adapted the Pivot algorithm and a LP rounding method for general weighted graphs; for the special probability weights case they showed that Pivot yields a 5-approximation and their LP rounding method yields a 2.5-approximation \cite{ailon2008aggregating}. Correlation clustering with probability weights has been adapted for probabilistic \cite{kollios2011clustering} and uncertain \cite{mandaglio2020and} graphs. On graphs with generalized weights, the best approximation ratio for correlation clustering is $O(\log n)$ \cite{charikar2005clustering,demaine2006correlation}. Puleo and Milenkovic also studied a partial generalization of graph weights \cite{puleo2015correlation}.

\subsubsection{The Pivot Algorithm} The Pivot algorithm was first introduced by Ailon et al. \cite{ailon2008aggregating}. Its efficient run time and ease of implementation have made it very popular, and it has been applied to many variants of correlation clustering that have arisen since. Recently, it has been used for uncertain graphs \cite{mandaglio2020and}, query-constrained correlation clustering \cite{garcia2020query}, online correlation clustering \cite{lattanzi2021robust}, chromatic correlation clustering \cite{klodt2021color}, and fair correlation clustering \cite{ahmadian2020fair}. It has also been shown how to run the Pivot algorithm in parallel in various settings \cite{chierichetti2014correlation,pan2015parallel}. Zuylen and Williamson \cite{van2009deterministic} developed a deterministic version of Pivot that picks a best pivot at each round, though at the cost of an increased running time complexity. The most efficient non-parallel implementation of Pivot uses a neighborhood oracle, where a hash table stores lists of neighbors for each node \cite{ailon2009correlation}.

\subsubsection{Consensus Clustering}

Though a problem of interest in its own right, consensus clustering has often been studied as a special case of correlation clustering \cite{gionis2007clustering, ailon2008aggregating, goder2008consensus}. An overview of consensus clustering methods is given by Vega-Pons and Ruiz-Shulcloper \cite{vega2011survey}.

\section{Weighted Correlation Clustering}\label{sec:weighted_cc}

Given a graph $G = (V, E)$, Ailon et al. \cite{ailon2008aggregating} considered a generalization of the correlation clustering problem by allowing every edge $(u, v)$ in $E$ to have a positive weight $w_{uv}^{+} \ge 0$ and negative weight $w_{uv}^{-} \ge 0$. A weighted correlation clustering instance $G$ has a corresponding unweighted \emph{majority} instance $G_{w}$; this is formed by adding edge $(i, j)$ to $E_{w}^{+}$ if $w_{ij}^{+} > w_{ij}^{-}$ and $(i, j)$ to $E_{w}^{-}$ if $w_{ij}^{-} > w_{ij}^{+}$ (breaking ties arbitrarily).

In this chapter we will assume the input graph satisfies the probability constraints $w_{uv}^{+} + w_{uv}^{-} = 1$. For edge $(u,v)$, we set $s(u,v) = w_{uv}^{+}$ and note that $1 - s(u,v) = w_{uv}^{-}$. The correlation clustering objective function is again given by 
$$\text{Cost}(\C, V) = \sum_{\substack{u, v \in V,\; u \ne v \\ (u,v) \text{ is intra-cluster}}} (1 - s(u,v)) + \sum_{\substack{u,v \in V,\; u \ne v \\ (u,v) \text{ is inter-cluster}}} s(u,v).$$
The standard correlation clustering problem arises when each $s(u,v)$ is equal to 0 or 1.

\subsection{Correlation Clustering Algorithms}\label{sec:weighted_cc_algs}

The Pivot algorithm runs directly on weighted graphs that satisfy the probability constraints \cite{kollios2011clustering}. The algorithm chooses a node $u$ at random, starts cluster $C = \{u\}$, and adds all other unclustered nodes $v$ to $C$ with weight $s(u,v) \ge 1/2$. It repeats on $V \setminus C$ until all nodes are clustered. The time complexity is $O(|V| + |E|)$, where $E$ is the complete set of edges between pairs of nodes in $V$. This can be improved to $O(|V| + |E^+|)$ where $E^+$ represents edges of weight $\ge 1/2$ if neighbor sets $N(v) = \{u \in V \mid s(u,v) \ge 1/2\}$ are known for all $v \in V$.

Ailon et al. showed that the Pivot algorithm yields a 5-approximation for the probability weights case. In the special case where the ``complements'' of these weights also satisfy a version of the triangle inequality (i.e. $1 - s(u,v) \le (1 - s(u,w)) + (1 - s(v,w))$ for $u, v, w \in V$), Ailon et al. further showed that the Pivot algorithm yields a 2-approximation.

Ailon et al. also presented a LP rounding method for weighted correlation clustering. We introduce a variable $x_{uv}$ for every pair of distinct nodes $u,v \in V$. We interpret $x_{uv} = 0$ to mean that $u$ and $v$ lie in the same cluster; $x_{uv} = 1$ means that $u$ and $v$ lie in different clusters. We assume $x_{uu} = 0$ always.
\begin{align}
\label{equation:lp_cc_weighted}
\begin{array}{lll}
\min & \sum\limits_{u,v \in V}[s(u,v) x_{uv} + (1 - s(u,v)) (1 - x_{uv})] & \\
\text{s. t.} & x_{uv} + x_{vw} \ge x_{uw} & \forall u, v, w \in V, \\
 & 0 \le x_{uv} \le 1 & \forall u, v \in V. 
\end{array}
\end{align}

Given a fractional solution to LP \ref{equation:lp_cc_weighted}, the rounding method from Ailon et al. yields a 2.5-approximation for the probability weights case, and a 2-approximation for probability weights that also satisfy the complement triangle inequality.

\subsection{Application to Consensus Clustering}

The input for consensus clustering is a set of clusterings $\C_{1} \dots, \C_{k}$ of a fixed set $V$. The goal is to output a new clustering $\C$ of $V$ that minimizes the \textit{disagreement distance} $\sum_{i = 1}^{k} \text{Disagree}(\C, \C_{i})$, where $\text{Disagree}(\C, \C_{i})$ is the number of node pairs $(i, j)$ that are clustered together in one clustering but not in the other.

Consensus clustering is closely related to weighted correlation clustering. We construct a graph by adding one node for every object in $V$, and then for every pair of objects $u, v \in V$ we add an edge $(u,v)$ with weight equal to the average number of clusterings where $u$ and $v$ appear in the same cluster. Constructed this way, the edge weights satisfy the complement triangle inequality \cite{ailon2008aggregating}.

The consensus clustering result is then given by the clustering formed by a correlation clustering algorithm on this weighted graph. The approximation bound of the correlation clustering algorithm equals the approximation bound for consensus clustering. Thus using either the Pivot or the LP rounding method presented in Section \ref{sec:weighted_cc_algs}, we get a 2-approximation algorithm for consensus clustering \cite{ailon2008aggregating}. From now on we will focus on the Pivot algorithm, since there are no theoretical gains from using the more expensive LP rounding method. In fact, the state-of-the-art method from Ailon el al. yields a 1.57-approximation by choosing the better clustering between the one produced by Pivot or choosing one of the input clusterings at random.

\section{Runtime Improvements}\label{sec:consensus_onthefly}

Previous implementations of the Pivot algoritm for consensus clustering have relied on precomputing the weighted graph used for these algorithms \cite{goder2008consensus}. For $k$ input clusterings on set $V$, precomputing all weights in the edge similarity graph requires $O(|V|^2 k)$ time and $O(|V|^2)$ space. Though correlation clustering algorithms like Pivot algorithm run fastest when these similarities are precomputed, the memory required to store edge weights quickly becomes unmanageable for larger graphs. Furthermore, the Pivot algorithm rarely uses all $O(|V|^2)$ edges during a single clustering run. Even if Pivot is run a handful of times, the extra cost of precomputing all possible edges still may not be justified. 

To achieve a run time improvement and to reduce overall memory usage, we can just compute and store cluster labels for each node for the $k$ input clusterings. A \textit{cluster label} is a $k$-tuple $(v_1, \dots, v_k)$ for a given node $v$, where $v_i$ is an integer denoting which cluster node $v$ is a member of in clustering $i$. This step runs in $O(|V|k)$ time and uses $O(|V|k)$ memory---the same amount of memory used by storing the input clusterings. Then computing the similarity between node $u$ and node $v$ is an $O(k)$ operation from counting the number of matching labels $v_i = u_i$. In practice the number of input clusterings is much smaller than $|V|$. By computing similarities ``on-the-fly'', the Pivot algorithm only incurs an extra factor of $k$ in its running time. 

We can use a similar approach for implementing other correlation clustering algorithms like LocalSearch \cite{gionis2007clustering} and Vote \cite{elsner2009bounding} by computing edge probabilities only as needed. This allows these algorithms to run on larger instances without requiring a great deal of extra memory to store edges. However, the run time gains of the ``on-the-fly'' approach is diminished for these algorithms if they are run multiple times, since both LocalSearch and Vote examine all graph edges during a single run. On smaller instances it is better to precompute edges in order to reuse them on multiple runs.

\section{Sampling Methods for Consensus Clustering}\label{sec:consensus_sampling}

Different sampling methods have been considered for correlation clustering \cite{bonchi2013local} and consensus clustering \cite{gionis2007clustering}, focusing on reducing the number of edge comparisons performed during clustering. Here we propose a method for consensus clustering to reduce the number of attributes used to calculate edge probabilities. 

\subsection{Pivot: Sampling Input Clusters}\label{sec:pivot_sampling}

For a large number $k$ of input clusterings, an additional factor of $O(k)$ in running time may not be practical. Let $R \in \{1, \dots, k\}$. We will analyze the effect of sampling only $R$ input clusters to compute edge probabilities.

For a given pair of nodes $i, j$, we can produce a $k$-bit string $s$ to represent which clusters they agree on. We set bit $l = 1$ if $i$ and $j$ are clustered together in clustering $l$, and 0 otherwise. Sampling a smaller number of clusterings is akin to drawing a smaller number of bits from this string. We will assume that $k$ is large so we can model these draws as sampling with replacement.

Let $p$ denote the true average number of clusters that $i$ and $j$ are clustered together in, and assume that $p \le 1/2$ (the opposite case will follow by symmetry). Let $X$ be the sum of $R$ randomly chosen bits of $s$. We model $X$ as a Binomial random variable with $R$ samples and success probability $p$. Since the Pivot algorithm will ``make a mistake'' when the sampled probability is $> 1/2$, we will first find the probability that $X > R/2$. 

To do this, we use the standard deviation $\sqrt{R p(1-p)}$ for $X$. Let $Z = (X - pR) / \sqrt{Rp(1-p)}$. We estimate $\PP(X > R/2)$ using
\begin{align*}
\PP(Z > (R/2 - pR) / \sqrt{R p(1-p)}) = \PP(Z > \sqrt{R}(1/2 - p) / \sqrt{p(1-p)}).  
\end{align*}
Let $f(R, p) = \sqrt{R}(1/2 - p) / \sqrt{p(1-p)}$. We find $\PP(Z > f(R, p))$ by evaluating $\text{Err}(R, p) := 1 - \Phi(f(R,p))$, where $\Phi$ is the normal CDF.

For $p < 1/2$, the cost of a correctly clustered edge is $p$. We are interested in how much this cost might increase due to an error in reading this probability from the sample. We will analyze the ``Node-at-a-Time'' approach to the Pivot algorithm.

As observed by Bonchi et al. \cite{bonchi2013local}, the Pivot can be adapted to run in $n$ rounds where a single node is clustered after each round. A permutation $\pi$ of the vertex set $V$ is fixed beforehand, and then nodes are processed in that order. A pivot node set $P$ is maintained. If an incoming node $v$ matches with pivot $i$ first, then $v$ is assigned cluster label $i$. If $v$ doesn't match with any nodes in $P$, then $v$ begins a new cluster and $v$ is added to $P$.

Given a permutation $\pi$ of $V$, we will break down what happens when Pivot ``makes a mistake'' into three cases:
\begin{enumerate}
    \item Node $v$ connects to the wrong pivot (but wasn't going to become a new pivot itself). In the worst case, the cost of all edges associated with $v$ will ``flip'' (becoming 1 - the original cost). This event happens with probability $\le \text{Err}(R, p)$, where $p$ is the true probability of the incorrectly sampled edge.
    
    \item Node $v$ connects to a pivot when it was supposed to become a new pivot. In this case we can create a new permutation $\pi '$ where node $v$ is sent to the back and all other nodes move up by one place. The cost increase is given in Case 1, but now relative to the new permutation $\pi'$ which also determines a Pivot clustering.
    
    \item Node $v$ becomes a new pivot when it was supposed to connect to an existing pivot. The cost incurred by node $v$ to already clustered nodes is given in Case 1. The cost incurred by subsequent nodes that connect to $v$ is given by Cases 1 and 2, depending on whether they should've connected to a different pivot or become a new pivot.
\end{enumerate}

In particular, we see that the expected cost due to error can be calculated based solely on edge probabilities. For a given $R$ and $p < 1/2$, we calculate the expected cost due to error as
\[p \cdot (1 - \text{Err}(R, p)) + (1 - p) \cdot \text{Err}(R, p).\]
(If $p \ge 1/2$ for a given edge, we replace $p$ with $1-p$ in the formula).

Fix a value for $R$. We want to find the maximum multiple of the original cost incurred from the sample. Define $g(R)$ to be
\[g(R) = \underset{p \in [0, 1/2]}{\max} \{ [p \cdot (1 - \text{Err}(R, p)) + (1 - p) \cdot \text{Err}(R, p)] / p\}.  \]

% Using only a small number of samples leads to a very small increase of cost:
% \begin{align*}
% \begin{array}{l|l|l|l|l|l}
% R & 50 & 100 & 200 & 500 & 1000  \\
% \hline 
% g(R) & 1.054 & 1.037 & 1.025 & 1.016 & 1.011
% \end{array}
% \end{align*}

Since the correlation clustering objective sums the cost of edges, and $g(R)$ is the multiple of the maximum expected increase of cost due to an error in sampling, we have the following

\begin{lemma}\label{lem:pivot_sampling}
For $R < k$, construct the consensus clustering similarity graph with $R$ randomly sampled input clusterings. The Pivot algorithm yields an expected $g(R) \cdot 2$-approximation when compared with the true similarity graph.
\end{lemma}

\subsection{Pivot Sampling Approximation Results}

The simplest algorithm for consensus clustering just picks an input clustering a random and returns it. This is noted to be a 2-approximation algorithm \cite{ailon2008aggregating}, so doing the extra work of computing edge probabilities and running the Pivot algorithm does not provide any theoretical benefit. Ailon et al. instead analyzed the following approach: run both of these algorithms, and return the result with lower cost. 

Let $t = (i, j, k)$ be a triangle of nodes. Let $w(t)$ be the worst-case cost clustering cost of a bad triangle $t$ from the Pivot algorithm, $z(t)$ be the cost of $t$ in the randomly chosen input clustering, and $c^{*}(t)$ be the optimal clustering cost of $t$. Ailon et al. proved the following 
\begin{theorem}
If there exist constants $\beta \in [0, 1]$ and $\gamma \ge 1$ such that $\beta w(t) + (1 - \beta) z(t) \le \gamma c^{*}(t)$ for all $t \in T$, then the best of Pivot and the random input clustering yields a $\gamma$-approximation for consensus clustering.
\end{theorem}

To find the particular value of $\gamma = 11/7$ using the Pivot algorithm, Ailon et al. show the following

\begin{lemma}
For all $t \in T$, $(3/7) w(t) + (4/7) z(t) \le (11/7) c^{*}(t)$.
\end{lemma}

\begin{proof}
To prove this, they show that
\[f(t) = (3/7) w(t) + (4/7) z(t) - (11/7) c^{*}(t) \le 0\]
where $t = (w_{1}, w_{2}, w_{3})$ and 
\begin{align*}
    w(t) &= w_{1} + w_{2} + w_{3} \\
    z(t) &= 2w_{1}(1 - w_{1}) + 2w_{2}(1 - w_{2}) + 2w_{3}(1 - w_{3}) \\
    c^{*}(t) &= w_{1} + 1 - w_{2} + 1 - w_{3} \\
    1/2 &\le w_{1} \le w_{j} \le 1 \text{ for } j = 2, 3 \\
    w_{1} &+ w_{2} + w_{3} \le 2
\end{align*}
They find a global maximum of $f(t)$ within the constrained area at $(w_{1}, w_{2}, w_{3}) = (1/2, 3/4, 3/4)$. Here $w(t) = 2$, $z(t) = 5/4$, and $c^{*}(t) = 1$, which yields $f(t) = 0$. So this bound is tight. 
\end{proof}

Using \lemref{lem:pivot_sampling}, we replace $w(t)$ with $g(R) \cdot w(t)$ for Pivot on $R$ inputs. Since $w(t)$ and $c^{*}(t)$ are linear functions, we can adjust their scalar multiples in $f(t)$ without affecting the location of the maximum value. Evaluating at $t^{*} = (w_{1}, w_{2}, w_{3}) = (1/2, 3/4, 3/4)$, we get
\begin{align*}
(3/7) &g(R) w(t^{*}) + (4/7) z(t^{*}) - \gamma c^{*}(t^{*}) \\
&= (3/7)(2) g(R) + (4/7)(5/4) - \gamma = (6/7) g(R) + 5/7 - \gamma.
\end{align*}
This is $\le 0$ when $\gamma \ge (6/7) g(R) + 5/7$. This yields the following
\begin{theorem}
The best of Pivot on $R$ input clusterings and picking a random input clustering yields a $[(6/7) g(R) + 5/7]$-approximation algorithm for consensus clustering.
\end{theorem}

The following graphs (Figure \ref{fig:consensus_bounds}) plot the consensus clustering approximation results for the Pivot algorithm. The graph on the left shows the value of the function $g(R)$ where $R$ is the number of samples used from the set of input clusterings; the baseline ``full'' is the constant value 1. The graph on the right shows the approximation bound based on $R$ samples, compared against the 11/7 baseline for sampling all inputs.
\begin{figure}[H]
    \centering
    \includegraphics[width=0.49\linewidth]{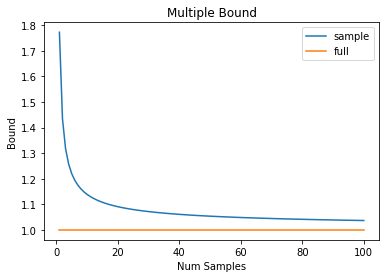}
    \includegraphics[width=0.49\linewidth]{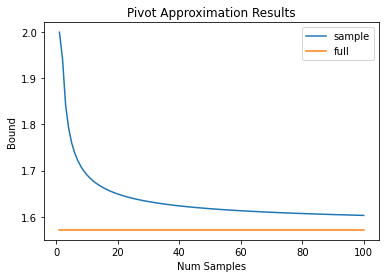}
    \caption{Consensus Clustering Theoretical Analysis}
    \label{fig:consensus_bounds}
\end{figure}

\section{Experiments}\label{sec:consensus_experiments}

All algorithms in this section were implemented in Java\footnote{code available at \texttt{github.com/nathan-cordner/cc-local-search}} and tested on a Linux server running Rocky Linux 8.7 with a 2.9 GHz processor and 16.2 GB of RAM. Plots show the mean result of 10 runs for each algorithm. Error bars show one standard deviation of the mean.

We compare the results of Pivot, Pivot with InnerLocalSearch (ILS---a LocalSearch method that only runs \textit{inside} individual Pivot clusters), Pivot with LocalSearch (LS), and Vote algorithms on each data set. Each algorithm computes similarities as needed using the approach outlined in Section \ref{sec:consensus_onthefly}. We restrict LocalSearch and InnerLocalSearch to just one iteration through the node set, since multiple iterations would require computing previously seen edge probabilities over again.

\subsection{Mushrooms Data}

The Mushrooms\footnote{\texttt{archive.ics.uci.edu/ml/datasets/mushroom}} dataset contains 8214 rows with 22 categorical attributes each (the original data contains 23 columns; the first column, which contains a classification into ``poisonous'' or ``edible'', is removed). We interpret each attribute as an input clustering: for a given attribute, all items with the same attribute value are put into one cluster. The Mushrooms data set has often been used for comparison in correlation \cite{garcia2020query} and consensus \cite{goder2008consensus} clustering experiments. 

\begin{figure}
    \centering
    \includegraphics[width=0.32\linewidth]{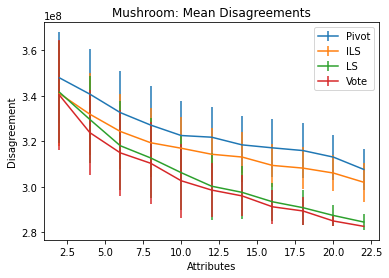}
    \includegraphics[width=0.32\linewidth]{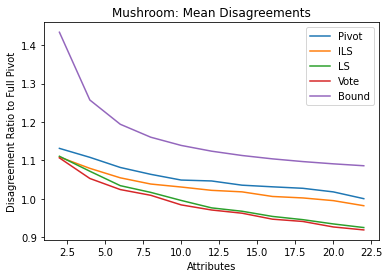}
    \includegraphics[width=0.32\linewidth]{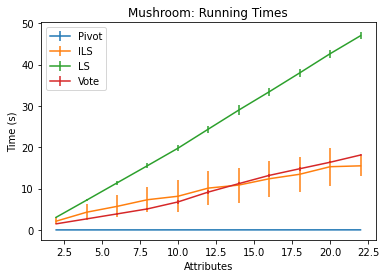}
    \caption{Mushroom Consensus Clustering}
    \label{fig:mushroom}
\end{figure}

We first compare the difference in running times of running Pivot on the fly versus precomputing all edge weights. Goder and Filkov \cite{goder2008consensus} used the precomputed version of Pivot for consensus clustering. Their experiments on the Mushrooms dataset yielded an average runtime of 1222 seconds, which they noted was ``dominated by the preprocessing.'' Our own time for preprocessing of the Mushrooms graph clocked in at 57.59 seconds, with an average run of Pivot being 0.0082 seconds afterward (out of 10 runs). However, running Pivot-on-the-Fly we computed the cluster labels in 0.029 seconds and had an average Pivot runtime of  0.0129 seconds afterward. Though each run of Pivot was marginally slower, we significantly reduce overall running time by not precomputing the entire similarity graph. For larger examples, storing precomputed edges becomes entirely infeasible.

In Figure \ref{fig:mushroom} we show the results of attribute sampling, where each algorithm computes edge probabilities as needed. The Pivot algorithm was run 50 times for each $R \in \{2, 4, \dots, 22\}$, where $R$ determines the number of randomly chosen input clusterings to use in the algorithm (with final disagreements computed against the full set of input clusterings). % The LocalSearch algorithms were then run once each for every Pivot consensus clustering. The mean disagreement costs and running times are plotted, with error bars showing one standard deviation from the mean. 

We first note the theoretical claims of Pivot sampling. For example, using only $R = 2$ inputs, our bound suggests that the Pivot algorithm should return a result less than 1.434 times the disagreement on the full set of inputs. Here we see that the mean Pivot clustering disagreement at $R = 2$ is only 1.131 times the disagreement for the full set---well below the theoretical bound. 

For the LocalSearch disagreement improvements, we note that both LS and ILS give similar results at $R = 2$ and begin to diverge as $R$ increases. The effectiveness of ILS diminishes as $R$ increases, whereas the full LS maintains its improvement levels. At $R = 22$ we see the largest gap between the two methods: the LS disagreement is 92.5\% of the Pivot level, whereas ILS is only 98.2\% of Pivot.

However, the running time graph is able to balance out the picture. For this graph, the Pivot running time is nearly instantaneous even as $R$ increases. The running time of the full LS increases the most dramatically, starting around 5 seconds for $R = 2$ and ending at nearly 50 seconds for a single pass over the data set when $R = 22$; note that the time used by LS for the full set of attributes is comparable to computing the full similarity set. The ILS method is much more efficient, clocking in at about 15 seconds for $R = 22$. LS takes over three times the amount of time as ILS, but for only a small improvement over the results of ILS.

We note that the Mushrooms data set is not ideal for ILS. The Pivot algorithm runs quickly since it forms a small number of clusters (about 10 for $R = 22$), which leads to larger cluster sizes (average max cluster size of about 3500 for $R = 22$). Thus ILS is not able to reduce the running time over LS as effectively, though we still see some improvements in disagreement cost and running time on this example.

The Vote algorithm performs quite well on the Mushrooms data set. Vote does not need to compute as many similarities as the full LocalSearch, so it runs much more efficiently (on the same level as ILS for this example). Vote also yields the greatest level of improvement over Pivot on all attribute samplings, edging out the improvements yielded by LS.

\subsection{Facebook Data} 

\begin{figure}
    \centering
    \includegraphics[width=0.32\linewidth]{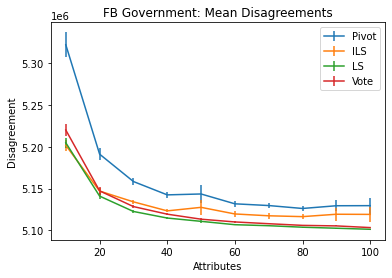}
    \includegraphics[width=0.32\linewidth]{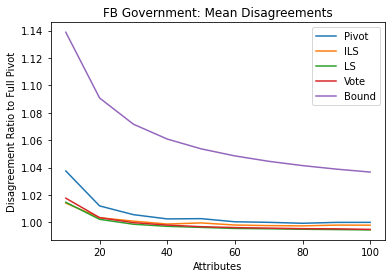}
    \includegraphics[width=0.32\linewidth]{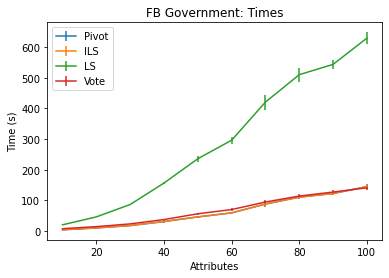}
    \caption{Facebook Government Consensus Clustering}
    \label{fig:fb_gov}
\end{figure}

The Facebook Government\footnote{\texttt{snap.stanford.edu/data/gemsec-Facebook.html}} graph is one of several different social networks between similarly themed Facebook pages. Nodes in the graph represent individual pages, and edges between nodes represent mutual likes between pages. The Government graph has 7,057 nodes and 89,455 edges. To perform consensus clustering, we first generate 100 Pivot clusterings and then run the consensus clustering algorithms on the results.

Figure \ref{fig:fb_gov} shows the disagreement cost and running time results of consensus clustering from our four algorithms. The Pivot algorithm was run 10 times for each of $R \in \{10, 20, \dots, 100\}$. 

Our theoretical approximation bound on $R = 10$ inputs suggests that the Pivot algorithm should return a result less than 1.139 times the disagreement on the full set. Here we see that the mean Pivot clustering disagreement at $R = 10$ is only 1.038 times the disagreement for the full set. For the LocalSearch disagreement improvements, we note that both LS and ILS give similar results at $R = 10$ and begin to diverge slightly as $R$ increases. At $R = 100$ however, the LS disagreement is only 99.45\% of Pivot with ILS being 99.8\% of Pivot. 

For running times, we see that the Pivot algorithm's time increases more dramatically as $R$ increases---starting with iterations under 5 seconds when $R = 10$, up to iterations of over 2 minutes when $R = 100$. This happens since the max cluster size $d$ of the Pivot algorithm stays small (under 50), and also means that the ILS algorithm runs extremely quickly with iterations lasting about a tenth of second even for $R = 100$. The LS algorithm, on the other hand, begins to perform quite poorly even on a small number of input clusterings. By $R = 50$ the LS iterations clock in at about 4 minutes, and are over 10 minutes at $R = 100$. This graph is small enough to precompute the entire similarity graph, which only takes about 5 minutes to complete. Since LS runs much longer than ILS, and does not give much added value over ILS in terms of disagreements reduced, it makes sense to favor ILS over LS for consensus clustering improvement.

Once again, the Vote algorithm performs relatively well on this graph. Since Pivot forms a large number of clusters, it has to compute most of the similarities between pairs of nodes in the graph. Thus the Vote algorithm runs in time comparable to Pivot and ILS (taking only slightly longer). We also see that Vote performs nearly as well as LS, giving the most significant improvement in disagreement costs with the least additional time required beyond Pivot.

\subsection{Correlated Binary Data}

\begin{figure*}
\caption{Correlated Binary Data Consensus Clustering}
\label{fig:binary_consensus}
\centering
  \begin{subfigure}[b]{0.32\linewidth}
  \centering 
  \includegraphics[width=\linewidth]{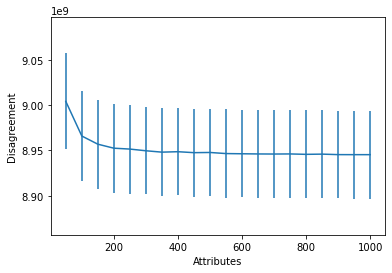}
  \caption{0.1 Mean 0.1 Correlation}
  \end{subfigure}
  \begin{subfigure}[b]{0.32\linewidth}
  \centering 
  \includegraphics[width=\linewidth]{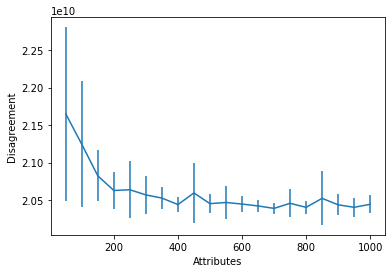}
  \caption{0.3 Mean 0.1 Correlation}
  \end{subfigure}
  \begin{subfigure}[b]{0.32\linewidth}
  \centering 
  \includegraphics[width=\linewidth]{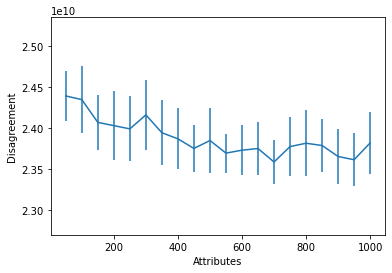}
  \caption{0.5 Mean 0.1 Correlation}
  \end{subfigure}
  \vspace{-6 pt}
  \begin{subfigure}[b]{0.32\linewidth}
  \centering 
  \includegraphics[width=\linewidth]{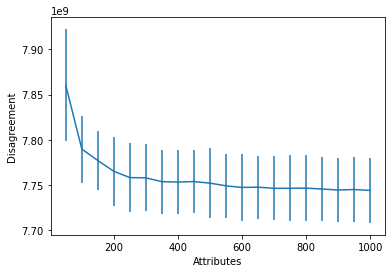}
  \caption{0.1 Mean 0.3 Correlation}
  \end{subfigure}
  \begin{subfigure}[b]{0.32\linewidth}
  \centering 
  \includegraphics[width=\linewidth]{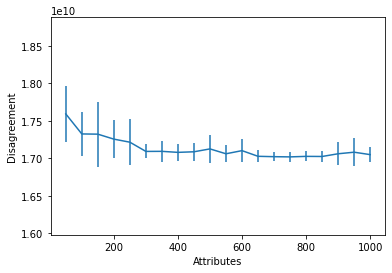}
  \caption{0.3 Mean 0.3 Correlation}
  \end{subfigure}
  \begin{subfigure}[b]{0.32\linewidth}
  \centering 
  \includegraphics[width=\linewidth]{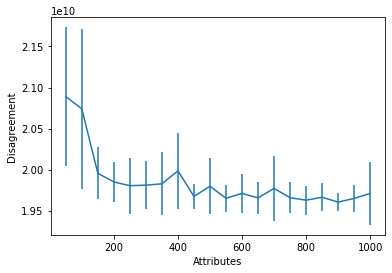}
  \caption{0.5 Mean 0.3 Correlation}
  \end{subfigure}
   \vspace{-6 pt}
    \begin{subfigure}[b]{0.32\linewidth}
  \centering 
  \includegraphics[width=\linewidth]{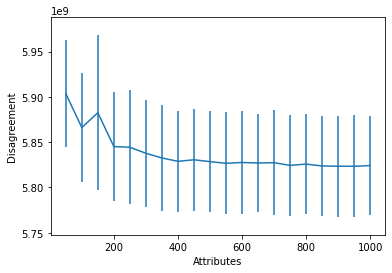}
  \caption{0.1 Mean 0.5 Correlation}
  \end{subfigure}
  \begin{subfigure}[b]{0.32\linewidth}
  \centering 
  \includegraphics[width=\linewidth]{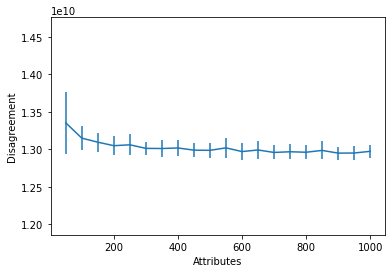}
  \caption{0.3 Mean 0.5 Correlation}
  \end{subfigure}
  \begin{subfigure}[b]{0.32\linewidth}
  \centering 
  \includegraphics[width=\linewidth]{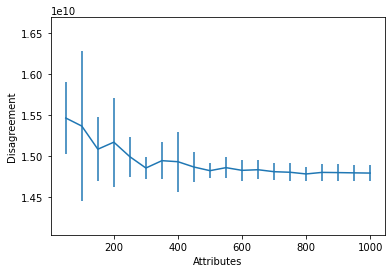}
  \caption{0.5 Mean 0.5 Correlation}
  \end{subfigure}
    \vspace{-6 pt}
    \begin{subfigure}[b]{0.32\linewidth}
  \centering 
  \includegraphics[width=\linewidth]{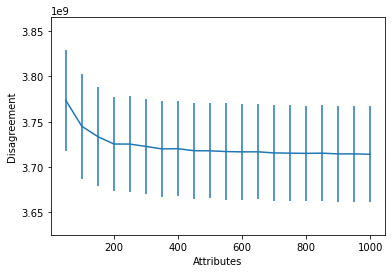}
  \caption{0.1 Mean 0.7 Correlation}
  \end{subfigure}
  \begin{subfigure}[b]{0.32\linewidth}
  \centering 
  \includegraphics[width=\linewidth]{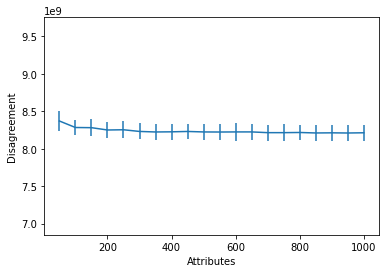}
  \caption{0.3 Mean 0.7 Correlation}
  \end{subfigure}
  \begin{subfigure}[b]{0.32\linewidth}
  \centering 
  \includegraphics[width=\linewidth]{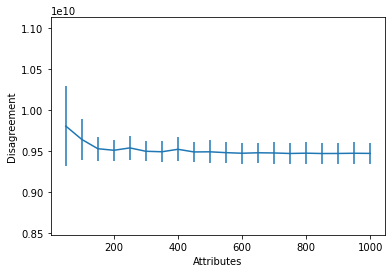}
  \caption{0.5 Mean 0.7 Correlation}
  \end{subfigure}
    \vspace{-6 pt}
    \begin{subfigure}[b]{0.325\linewidth}
  \centering 
  \includegraphics[width=\linewidth]{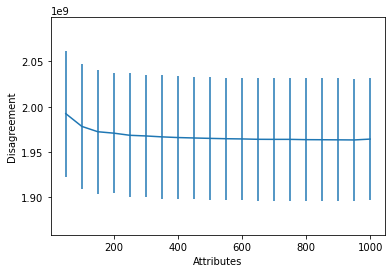}
  \caption{0.1 Mean 0.9 Correlation}
  \end{subfigure}
  \begin{subfigure}[b]{0.32\linewidth}
  \centering 
  \includegraphics[width=\linewidth]{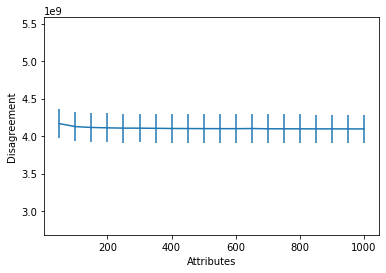}
  \caption{0.3 Mean 0.9 Correlation}
  \end{subfigure}
  \begin{subfigure}[b]{0.32\linewidth}
  \centering 
  \includegraphics[width=\linewidth]{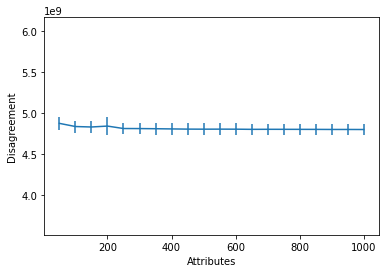}
  \caption{0.5 Mean 0.9 Correlation}
  \end{subfigure}
\end{figure*}

\begin{figure*}
\caption{Correlated Binary Data Ratios}
\label{fig:binary_consensus_multiples}
\centering
  \begin{subfigure}[b]{0.32\linewidth}
  \centering 
  \includegraphics[width=\linewidth]{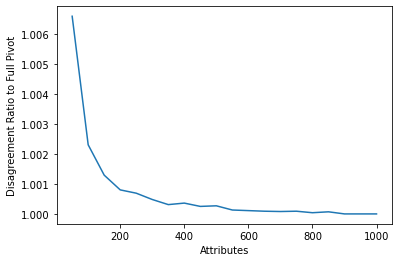}
  \caption{0.1 Mean 0.1 Correlation}
  \end{subfigure}
  \begin{subfigure}[b]{0.32\linewidth}
  \centering 
  \includegraphics[width=\linewidth]{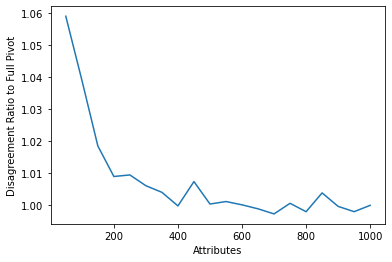}
  \caption{0.3 Mean 0.1 Correlation}
  \end{subfigure}
  \begin{subfigure}[b]{0.32\linewidth}
  \centering 
  \includegraphics[width=\linewidth]{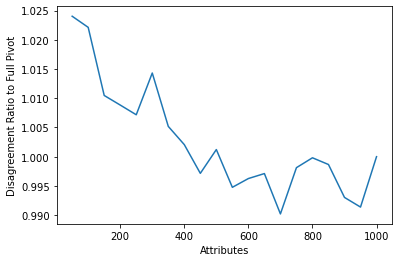}
  \caption{0.5 Mean 0.1 Correlation}
  \end{subfigure}
  \begin{subfigure}[b]{0.32\linewidth}
  \centering 
  \includegraphics[width=\linewidth]{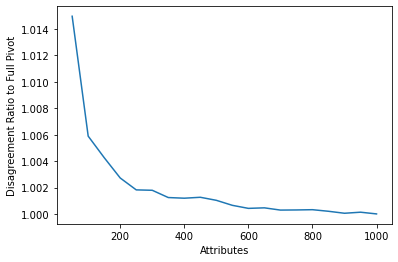}
  \caption{0.1 Mean 0.3 Correlation}
  \end{subfigure}
  \begin{subfigure}[b]{0.32\linewidth}
  \centering 
  \includegraphics[width=\linewidth]{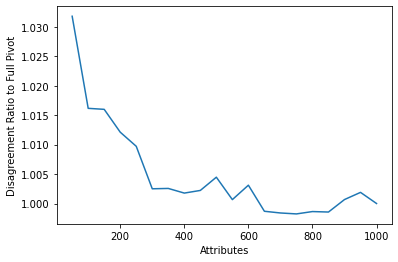}
  \caption{0.3 Mean 0.3 Correlation}
  \end{subfigure}
  \begin{subfigure}[b]{0.32\linewidth}
  \centering 
  \includegraphics[width=\linewidth]{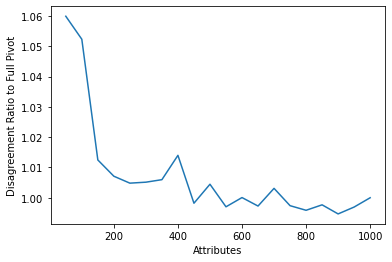}
  \caption{0.5 Mean 0.3 Correlation}
  \end{subfigure}
    \begin{subfigure}[b]{0.32\linewidth}
  \centering 
  \includegraphics[width=\linewidth]{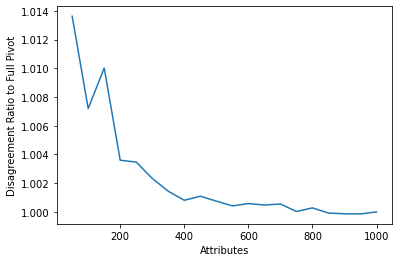}
  \caption{0.1 Mean 0.5 Correlation}
  \end{subfigure}
  \begin{subfigure}[b]{0.32\linewidth}
  \centering 
  \includegraphics[width=\linewidth]{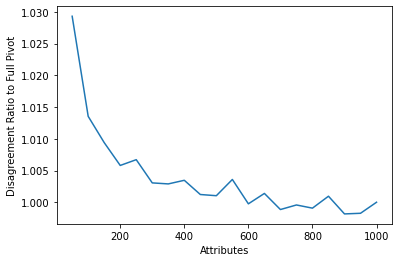}
  \caption{0.3 Mean 0.5 Correlation}
  \end{subfigure}
  \begin{subfigure}[b]{0.32\linewidth}
  \centering 
  \includegraphics[width=\linewidth]{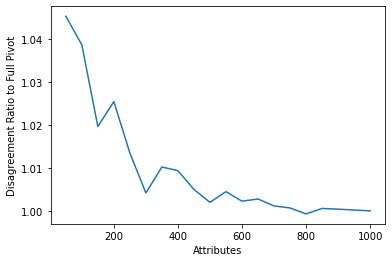}
  \caption{0.5 Mean 0.5 Correlation}
  \end{subfigure}
    \begin{subfigure}[b]{0.32\linewidth}
  \centering 
  \includegraphics[width=\linewidth]{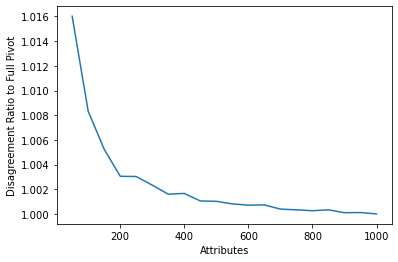}
  \caption{0.1 Mean 0.7 Correlation}
  \end{subfigure}
  \begin{subfigure}[b]{0.32\linewidth}
  \centering 
  \includegraphics[width=\linewidth]{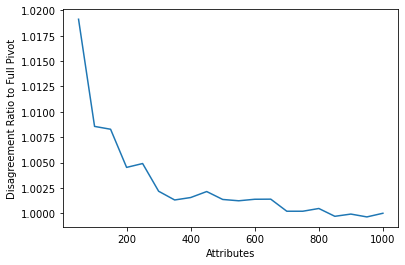}
  \caption{0.3 Mean 0.7 Correlation}
  \end{subfigure}
  \begin{subfigure}[b]{0.32\linewidth}
  \centering 
  \includegraphics[width=\linewidth]{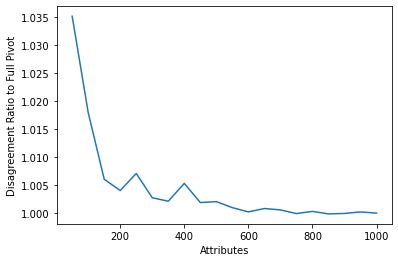}
  \caption{0.5 Mean 0.7 Correlation}
  \end{subfigure}
    \begin{subfigure}[b]{0.325\linewidth}
  \centering 
  \includegraphics[width=\linewidth]{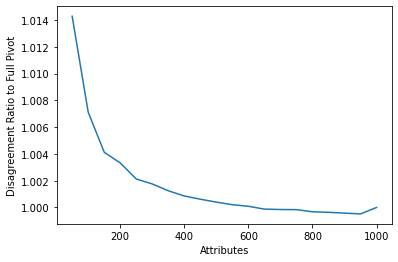}
  \caption{0.1 Mean 0.9 Correlation}
  \end{subfigure}
  \begin{subfigure}[b]{0.32\linewidth}
  \centering 
  \includegraphics[width=\linewidth]{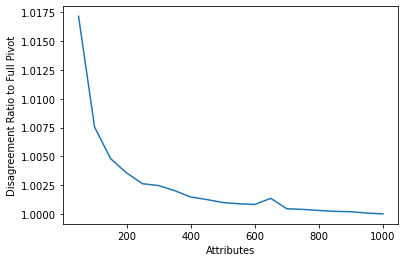}
  \caption{0.3 Mean 0.9 Correlation}
  \end{subfigure}
  \begin{subfigure}[b]{0.32\linewidth}
  \centering 
  \includegraphics[width=\linewidth]{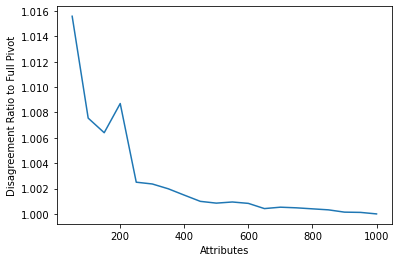}
  \caption{0.5 Mean 0.9 Correlation}
  \end{subfigure}
\end{figure*}

\iffalse 
\begin{figure*}
\caption{Correlated Binary Data Consensus Clustering}
\label{fig:binary_consensus}
\centering
  \begin{subfigure}[b]{0.32\linewidth}
  \centering 
  \includegraphics[width=\linewidth]{1_1.png}
  \caption{0.1 Mean, 0.1 Correlation}
  \end{subfigure}
  \begin{subfigure}[b]{0.32\linewidth}
  \centering 
  \includegraphics[width=\linewidth]{1_5.png}
  \caption{0.1 Mean, 0.5 Correlation}
  \end{subfigure}
  \begin{subfigure}[b]{0.32\linewidth}
  \centering 
  \includegraphics[width=\linewidth]{1_9.png}
  \caption{0.1 Mean, 0.9 Correlation}
  \end{subfigure}
  \begin{subfigure}[b]{0.32\linewidth}
  \centering 
  \includegraphics[width=\linewidth]{3_1.png}
  \caption{0.3 Mean, 0.1 Correlation}
  \end{subfigure}
  \begin{subfigure}[b]{0.32\linewidth}
  \centering 
  \includegraphics[width=\linewidth]{3_5.png}
  \caption{0.3 Mean, 0.5 Correlation}
  \end{subfigure}
  \begin{subfigure}[b]{0.32\linewidth}
  \centering 
  \includegraphics[width=\linewidth]{3_9.png}
  \caption{0.3 Mean, 0.9 Correlation}
  \end{subfigure}
    \begin{subfigure}[b]{0.32\linewidth}
  \centering 
  \includegraphics[width=\linewidth]{5_1.png}
  \caption{0.5 Mean, 0.1 Correlation}
  \end{subfigure}
  \begin{subfigure}[b]{0.32\linewidth}
  \centering 
  \includegraphics[width=\linewidth]{5_5.png}
  \caption{0.5 Mean, 0.5 Correlation}
  \end{subfigure}
  \begin{subfigure}[b]{0.32\linewidth}
  \centering 
  \includegraphics[width=\linewidth]{5_9.png}
  \caption{0.5 Mean, 0.9 Correlation}
  \end{subfigure}
\end{figure*}
\fi 

The approximation bound analysis in Section \ref{sec:pivot_sampling} for the Pivot sampling algorithm relied on the independence of attributes in the data set. Here we show experimentally that sampling still provides good results even if attributes are correlated.

We generated binary data sets with $|V| = 10000$ and $k = 1000$ using \texttt{bindata} in R \cite{bindataR, leisch1998generation}. All attributes are drawn using a fixed marginal probability and a fixed pairwise correlation, and 5 data sets are generated for each probability-correlation pair. For each data set we ran just the Pivot algorithm 10 times for each of $R\in \{50, \dots, 1000\}$ (increasing by 50 each time), tracking the disagreement distance from the resulting clustering to the input clusterings.

For each data set, we rescale outlier disagreements that are more than two times the median distance from the median to that boundary value. We plot the averages of the disagreements across the 5 data sets for each probability-correlation pair. We do not report the running times since even the longest runs finish in about 5 seconds. The exact disagreements are displayed in Figure \ref{fig:binary_consensus}. For comparison, we fix the $y$-axis size of the graphs for each mean. The ratios to Pivot on the full set of attributes given in Figure \ref{fig:binary_consensus_multiples}.

As expected, we see that the lines get flatter as the correlation increases. For example, the data sets with mean 0.3 start with a ratio of $1.059$ at $R = 50$ with correlation 0.1, then decreases to 1.032 (correlation 0.3), 1.029 (correlation 0.5), 1.019 (correlation 0.7), and 1.017 (correlation 0.9). The more correlated the attributes become, the more information we gain from smaller sample sets. 

In general the multiplicative ratios reported in the experiments stay within the theoretical bounds. There are a few minor exceptions, due to the averaging and rescaling of results from different data sets. For example, the value of $g(R)$ is 1.054 for $R = 50$ and 1.037 for $R = 100$. The result for 0.3 mean, 0.1 correlation reports a multiplicative ratio of 1.059 for $R = 50$ and 1.039 for $R = 100$ between the mean disagreement of the sample to the mean disagreement of the full set of clusterings. The 0.5 mean, 0.3 correlation and the 0.5 mean, 0.5 correlation graphs also have two bound violations each. The rest of the bounds on these data sets and the rest of the binary data sets fit within the predicted theoretical bounds.

\section{Conclusion}

In this paper we examined the consensus clustering problem from the standpoint of correlation clustering. We overcame some significant memory roadblocks to running the Pivot algorithm on larger data sets, and showed the viability of running other correlation clustering algorithms like LocalSearch, InnerLocalSearch, and Vote for consensus clustering too. We also examined the possibility of performing consensus clustering on limited samples of input clusterings, and showed that we can still obtain quality clustering results even while using relatively small sample sets.

\bibliography{sources}

\begin{thebibliography}{27}
\providecommand{\natexlab}[1]{#1}
\providecommand{\url}[1]{\texttt{#1}}
\expandafter\ifx\csname urlstyle\endcsname\relax
  \providecommand{\doi}[1]{doi: #1}\else
  \providecommand{\doi}{doi: \begingroup \urlstyle{rm}\Url}\fi

\bibitem[Ahmadian et~al.(2020)Ahmadian, Epasto, Kumar, and
  Mahdian]{ahmadian2020fair}
S.~Ahmadian, A.~Epasto, R.~Kumar, and M.~Mahdian.
\newblock Fair correlation clustering.
\newblock In \emph{International Conference on Artificial Intelligence and
  Statistics}, pages 4195--4205. PMLR, 2020.

\bibitem[Ailon and Liberty(2009)]{ailon2009correlation}
N.~Ailon and E.~Liberty.
\newblock Correlation clustering revisited: The “true” cost of error
  minimization problems.
\newblock In \emph{International Colloquium on Automata, Languages, and
  Programming}, pages 24--36. Springer, 2009.

\bibitem[Ailon et~al.(2008)Ailon, Charikar, and Newman]{ailon2008aggregating}
N.~Ailon, M.~Charikar, and A.~Newman.
\newblock Aggregating inconsistent information: ranking and clustering.
\newblock \emph{Journal of the ACM (JACM)}, 55\penalty0 (5):\penalty0 1--27,
  2008.

\bibitem[Bansal et~al.(2004)Bansal, Blum, and Chawla]{bansal2004correlation}
N.~Bansal, A.~Blum, and S.~Chawla.
\newblock Correlation clustering.
\newblock \emph{Machine learning}, 56\penalty0 (1-3):\penalty0 89--113, 2004.

\bibitem[Bonchi et~al.(2013)Bonchi, Garc{\'\i}a-Soriano, and
  Kutzkov]{bonchi2013local}
F.~Bonchi, D.~Garc{\'\i}a-Soriano, and K.~Kutzkov.
\newblock Local correlation clustering.
\newblock \emph{arXiv preprint arXiv:1312.5105}, 2013.

\bibitem[Bonchi et~al.(2014)Bonchi, Garcia-Soriano, and
  Liberty]{bonchi2014correlation}
F.~Bonchi, D.~Garcia-Soriano, and E.~Liberty.
\newblock Correlation clustering: from theory to practice.
\newblock In \emph{KDD}, page 1972, 2014.

\bibitem[Charikar et~al.(2005)Charikar, Guruswami, and
  Wirth]{charikar2005clustering}
M.~Charikar, V.~Guruswami, and A.~Wirth.
\newblock Clustering with qualitative information.
\newblock \emph{Journal of Computer and System Sciences}, 71\penalty0
  (3):\penalty0 360--383, 2005.

\bibitem[Chierichetti et~al.(2014)Chierichetti, Dalvi, and
  Kumar]{chierichetti2014correlation}
F.~Chierichetti, N.~Dalvi, and R.~Kumar.
\newblock Correlation clustering in mapreduce.
\newblock In \emph{Proceedings of the 20th ACM SIGKDD international conference
  on Knowledge discovery and data mining}, pages 641--650, 2014.

\bibitem[Cohen-Addad et~al.(2022)Cohen-Addad, Lee, and
  Newman]{cohen2022correlation}
V.~Cohen-Addad, E.~Lee, and A.~Newman.
\newblock Correlation clustering with sherali-adams.
\newblock In \emph{2022 IEEE 63rd Annual Symposium on Foundations of Computer
  Science (FOCS)}, pages 651--661. IEEE, 2022.

\bibitem[Demaine et~al.(2006)Demaine, Emanuel, Fiat, and
  Immorlica]{demaine2006correlation}
E.~D. Demaine, D.~Emanuel, A.~Fiat, and N.~Immorlica.
\newblock Correlation clustering in general weighted graphs.
\newblock \emph{Theoretical Computer Science}, 361\penalty0 (2-3):\penalty0
  172--187, 2006.

\bibitem[Elsner and Schudy(2009)]{elsner2009bounding}
M.~Elsner and W.~Schudy.
\newblock Bounding and comparing methods for correlation clustering beyond ilp.
\newblock In \emph{Proceedings of the Workshop on Integer Linear Programming
  for Natural Language Processing}, pages 19--27, 2009.

\bibitem[Garc{\'\i}a-Soriano et~al.(2020)Garc{\'\i}a-Soriano, Kutzkov, Bonchi,
  and Tsourakakis]{garcia2020query}
D.~Garc{\'\i}a-Soriano, K.~Kutzkov, F.~Bonchi, and C.~Tsourakakis.
\newblock Query-efficient correlation clustering.
\newblock In \emph{Proceedings of The Web Conference 2020}, pages 1468--1478,
  2020.

\bibitem[Gionis et~al.(2007)Gionis, Mannila, and
  Tsaparas]{gionis2007clustering}
A.~Gionis, H.~Mannila, and P.~Tsaparas.
\newblock Clustering aggregation.
\newblock \emph{Acm transactions on knowledge discovery from data (tkdd)},
  1\penalty0 (1):\penalty0 4--es, 2007.

\bibitem[Goder and Filkov(2008)]{goder2008consensus}
A.~Goder and V.~Filkov.
\newblock Consensus clustering algorithms: Comparison and refinement.
\newblock In \emph{2008 Proceedings of the Tenth Workshop on Algorithm
  Engineering and Experiments (ALENEX)}, pages 109--117. SIAM, 2008.

\bibitem[Haruna et~al.(2018)Haruna, Hou, Eghan, Kpiebaareh, and
  Tandoh]{haruna2018hybrid}
C.~R. Haruna, M.~Hou, M.~J. Eghan, M.~Y. Kpiebaareh, and L.~Tandoh.
\newblock A hybrid data deduplication approach in entity resolution using
  chromatic correlation clustering.
\newblock In \emph{International Conference on Frontiers in Cyber Security},
  pages 153--167. Springer, 2018.

\bibitem[Klodt et~al.(2021)Klodt, Seifert, Zahn, Casel, Issac, and
  Friedrich]{klodt2021color}
N.~Klodt, L.~Seifert, A.~Zahn, K.~Casel, D.~Issac, and T.~Friedrich.
\newblock A color-blind 3-approximation for chromatic correlation clustering
  and improved heuristics.
\newblock In \emph{Proceedings of the 27th ACM SIGKDD Conference on Knowledge
  Discovery \& Data Mining}, pages 882--891, 2021.

\bibitem[Kollios et~al.(2011)Kollios, Potamias, and
  Terzi]{kollios2011clustering}
G.~Kollios, M.~Potamias, and E.~Terzi.
\newblock Clustering large probabilistic graphs.
\newblock \emph{IEEE Transactions on Knowledge and Data Engineering},
  25\penalty0 (2):\penalty0 325--336, 2011.

\bibitem[Lattanzi et~al.(2021)Lattanzi, Moseley, Vassilvitskii, Wang, and
  Zhou]{lattanzi2021robust}
S.~Lattanzi, B.~Moseley, S.~Vassilvitskii, Y.~Wang, and R.~Zhou.
\newblock Robust online correlation clustering.
\newblock \emph{Advances in Neural Information Processing Systems}, 34, 2021.

\bibitem[Leisch et~al.(1998)Leisch, Weingessel, and
  Hornik]{leisch1998generation}
F.~Leisch, A.~Weingessel, and K.~Hornik.
\newblock On the generation of correlated artificial binary data.
\newblock \emph{Working Paper Series, SFB ``Adaptive Information Systems and
  Modelling in Economics and Management Science''}, 1998.

\bibitem[Leisch et~al.(2021)Leisch, Weingessel, and Hornik]{bindataR}
F.~Leisch, A.~Weingessel, and K.~Hornik.
\newblock \emph{bindata: Generation of Artificial Binary Data}, 2021.
\newblock URL \url{https://CRAN.R-project.org/package=bindata}.
\newblock R package version 0.9-20.

\bibitem[Mandaglio et~al.(2020)Mandaglio, Tagarelli, and
  Gullo]{mandaglio2020and}
D.~Mandaglio, A.~Tagarelli, and F.~Gullo.
\newblock In and out: Optimizing overall interaction in probabilistic graphs
  under clustering constraints.
\newblock In \emph{Proceedings of the 26th ACM SIGKDD International Conference
  on Knowledge Discovery \& Data Mining}, pages 1371--1381, 2020.

\bibitem[Pan et~al.(2015)Pan, Papailiopoulos, Oymak, Recht, Ramchandran, and
  Jordan]{pan2015parallel}
X.~Pan, D.~Papailiopoulos, S.~Oymak, B.~Recht, K.~Ramchandran, and M.~I.
  Jordan.
\newblock Parallel correlation clustering on big graphs.
\newblock In \emph{Advances in Neural Information Processing Systems}, pages
  82--90, 2015.

\bibitem[Puleo and Milenkovic(2015)]{puleo2015correlation}
G.~J. Puleo and O.~Milenkovic.
\newblock Correlation clustering with constrained cluster sizes and extended
  weights bounds.
\newblock \emph{SIAM Journal on Optimization}, 25\penalty0 (3):\penalty0
  1857--1872, 2015.

\bibitem[Shi et~al.(2021)Shi, Dhulipala, Eisenstat, Lacki, and
  Mirrokni]{shi2021scalable}
J.~Shi, L.~Dhulipala, D.~Eisenstat, J.~Lacki, and V.~Mirrokni.
\newblock Scalable community detection via parallel correlation clustering.
\newblock \emph{Proceedings of the VLDB Endowment}, 14\penalty0 (11):\penalty0
  2305--2313, 2021.

\bibitem[Van~Zuylen and Williamson(2009)]{van2009deterministic}
A.~Van~Zuylen and D.~P. Williamson.
\newblock Deterministic pivoting algorithms for constrained ranking and
  clustering problems.
\newblock \emph{Mathematics of Operations Research}, 34\penalty0 (3):\penalty0
  594--620, 2009.

\bibitem[Vega-Pons and Ruiz-Shulcloper(2011)]{vega2011survey}
S.~Vega-Pons and J.~Ruiz-Shulcloper.
\newblock A survey of clustering ensemble algorithms.
\newblock \emph{International Journal of Pattern Recognition and Artificial
  Intelligence}, 25\penalty0 (03):\penalty0 337--372, 2011.

\bibitem[Veldt et~al.(2018)Veldt, Gleich, and Wirth]{veldt2018correlation}
N.~Veldt, D.~F. Gleich, and A.~Wirth.
\newblock A correlation clustering framework for community detection.
\newblock In \emph{Proceedings of the 2018 World Wide Web Conference}, pages
  439--448, 2018.

\end{thebibliography}

\end{document}